\documentclass[12pt]{article}%
\usepackage{amssymb,amsfonts,amsmath,amsthm}
\usepackage{color}
\usepackage{cite}
\usepackage{blkarray}
\usepackage{graphicx}%
\providecommand{\U}[1]{\protect\rule{.1in}{.1in}}
\newtheorem{theorem}{Theorem}[section]
\newtheorem{proposition}[theorem]{Proposition}
\newtheorem{definition}[theorem]{Definition}
\newtheorem{lemma}[theorem]{Lemma}

\newtheorem{example}[theorem]{Example}

\def\F{\mathbb F}
\def\R{\mathbb R}
\def\cX{\mathcal X}
\def\cY{\mathcal Y}
\DeclareMathOperator*{\argmin}{arg\,min}
\DeclareMathOperator*{\argmax}{arg\,max}

\begin{document}

\title{Matched Metrics and Channels}
\author{Marcelo Firer and Judy L.\ Walker}
\date{\today}
\maketitle

\begin{abstract}
  The most common decision criteria for decoding are maximum
  likelihood decoding and nearest neighbor decoding. It is well-known
  that maximum likelihood decoding coincides with nearest neighbor
  decoding with respect to the Hamming metric on the binary symmetric
  channel. In this work we study channels and metrics for which those
  two criteria do and do not coincide for general codes.
\end{abstract}

\section{Introduction}

In coding theory, the Hamming metric has a prominent status, since it
can be used to perform maximum likelihood (ML) decoding over a
memoryless binary symmetric channel (BSC), in the sense that decoding
by choosing a most probable codeword (ML decoding)  is actually the
same decision as decoding by choosing a closest codeword (nearest
neighbor (NN) decoding). In this decision criteria sense, we have also
that Euclidean distance is the proper distance for modulation-decoding
when considering a continuous channel with white Gaussian noise (see,
e.g., \cite{Costello}) and the Lee metric has the same distinguished
role when considering some kinds of modulation and transmission over
certain discrete memoryless $q$-ary channels \cite{ChiangWolf}.

The use of geometric properties of channels in coding theory is
explored in many generic situations, such as the one proposed by
Forney \cite{Forney} for geometric uniformity of codes on
continuous channels and the study of geometrically inspired properties
of codes over discrete channels, as in \cite{FVY} and \cite{KuKi},
where bounds for the packing radius are derived from a distance-like
structure defined on a hypergraph determined by the channel model.

Many different distances are considered in the context of coding
theory (a comprehensive account may be found in \cite[Chapter
16]{Deza}), but not much is known about the general relationship between
channel models and metrics and not much is known about the geometry of
many important channels.

In this work we are concerned with the most basic of those questions:
is any ML decoder also an NN decoder, and conversely, is any NN
decoder also an ML decoder? More precisely, if, for every code, ML
decoding on a given channel coincides with NN decoding with respect to
a given metric, we say the channel and the metric are \emph{matched}
to one another.

This terminology goes back over 40 years, as a 1971 paper
\cite{ChiangWolf} attributes it to notes from a 1967 course given by
Massey \cite{Massey}.  In \cite{ChiangWolf}, Chiang and Wolf classify
the channels matched to the Lee metric.  Their results are generalized
in a 1980 paper by S\'eguin \cite{Seguin}, which studies necessary and
sufficient conditions for a discrete memoryless channel to be matched
to an \emph{additive} metric, i.e., a metric that is defined on an
alphabet $A$ and then extended to a metric on $A^n$ by applying the
metric coordinate-wise and taking the sum.  Gabidulin returns to this
question in a 2007 book chapter \cite{Gabidulin} and asserts -- with
an implicit assumption that all metrics under consideration are
additive -- that matched metrics exist only ``for very restricted
channels'' and then studies a weakening of the matching condition.
Finally, to avoid potential confusion by the reader, we also mention
recent work on \emph{mismatched decoders} (see, e.g., \cite{Ziv}) but
note that that work considers different questions than those that are
considered here.  In particular, given a channel with transition
probabilities $\Pr(x | y)$, the quantity $-\frac1{n} \log \Pr(x | y)$
considered in \cite{Ziv} does not, in general, determine a metric in
any sense.

This work is organized as follows: In Section \ref{mas} we introduce
the rigorous definition of the matching problem and we show that any
metric admits a matched channel; we also show by example that the
converse does not always hold. We therefore also give some conditions
for a channel that obstruct the existence of a matched metric. In
Section \ref{Z} we construct a matched metric for the $Z$-channel. In
Section \ref{bac} we conjecture that any binary asymmetric channel
(BAC) admits a matched metric and present some evidence for this
conjecture.
%
%
%
%

\section{Matched metrics and channels}\label{mas}

It is well-known that, under the assumption of equally likely
codewords, maximum likelihood decoding coincides with nearest neighbor
decoding with respect to the Hamming metric on the binary symmetric
channel.  This is a general fact that does not depend on the code and
so we ask: for what other channels is there such a metric?  Following
Massey \cite{Massey}, we call such a channel-metric pairs
\emph{matched}, a term we shall define rigorously below. We first
recall that a \emph{channel} $W$ with input and output alphabets $\cX$ is
defined by a conditional probability distribution $\Pr:\cX \times \cX
\to \R$, where
\[
\Pr ( x|y) = \Pr(x \textrm{ received} \, |\, y \textrm{ sent}).
\]
A \emph{metric} on $\cX$ is a function $d:\cX \times \cX \to \R$ such
that:
\begin{enumerate}
\item $d$ is \emph{symmetric}: $d(x,y) = d(y,x)$ for all $x,y \in
  \cX$;
\item $d$ is \emph{nonnegative}: $d(x,y) \geq 0$ for all $x,y \in
  \cX$, with equality if and only if $x=y$; and
\item $d$ satisfies the \emph{triangle inequality}: $d(x,y) + d(y,z)
  \geq d(x,z)$ for all $x,y,z \in \cX$.
\end{enumerate}

We may now give a rigorous definition of a matched pair, as follows:

\begin{definition}
  Let $W:\cX \to \cX$ be a channel with input and output alphabets
  $\cX$ and let $d$ be a metric on $\cX$.
We say that $W$ and $d$ are \emph{matched} if maximum likelihood
decoding on $W$ coincides with nearest neighbor decoding with respect
to $d$ for every code $C\subseteq\cX$, i.e., if for every code
$C\subseteq \cX$ and every $x \in \cX$, we have
\begin{equation}\label{matched}
\argmax_{y \in C} \Pr(x \textrm{ received}\, |\, y \textrm{ sent}) =
\argmin_{y \in C} d(x,y).
\end{equation}
\end{definition}

Several comments are in order.  First, note that we are not
restricting to additive channels or metrics.  That is, we are considering
codes to be subsets of the alphabet, so that the binary symmetric
channel, for example, should be considered as a channel on $\F_2^n$
rather than as $n$ uses of a channel defined on $\F_2$, and the Hamming
metric should be considered as a metric on $\F_2^n$ rather than as a sum
of $n$ copies of the Hamming metric on $\F_2$.  Next, by considering
codes with just two codewords, it is straightforward that condition
(\ref{matched}) is equivalent to the condition that, for every $x$,
$y$, $z \in \cX$ with either $\Pr(x\text{ received}\,|\,y\text{
  sent})>0$ or $\Pr(x\text{ received}\,|\,z\text{ sent})>0$ (or both),
\begin{equation}\label{newmatched}
\Pr(x\text{ received}\,|\,y\text{ sent})>\Pr(x\text{ received}\,|\,z\text{
sent})\text{ \quad if and only if \quad }d(x,y)<d(x,z).
\end{equation}
Finally, we make the following two assumptions throughout the paper:
\begin{itemize}
\item Every channel is \emph{reasonable} in the sense that 
\[
\Pr(x
  \textrm{ received}\, |\, x \textrm{ sent}) > \Pr(x \textrm{
    received} \, | \, y \textrm{ sent})
\]
for all $x \neq y \in \cX$. 
\item All codewords are equally likely.
\end{itemize}
The first assumption is a necessary condition for a channel to admit a
matched metric since $0 = d(x,x)<d(x,y)$ for all $x \neq y \in \cX$
and any metric $d$.  The second condition is reasonable, since it does
not depend on the channel but rather on the messages to be sent; it is
needed in order for maximum likelihood decoding to be relevant.
(Alternatively, one could drop this assumption and replace ``maximum
likelihood decoding'' with ``maximum a posteriori probability
decoding'' throughout the paper.)  In light of this second assumption,
we note that, by Bayes' theorem (see, e.g., \cite{probtext}),
conditions~(\ref{matched}) and (\ref{newmatched}) are also equivalent
to
\begin{equation}\label{newnewmatched}
\Pr(y \text{ sent} \, | \, x \text{ received}) > \Pr(z \text{ sent} \,
| \, x \text{ received}) \text{ \quad if and only if \quad }
d(x,y)<d(x,z)
\end{equation}
for every $x,y,z \in \cX$ with either $\Pr(x\text{ received}\,|\,y\text{
  sent})>0$ or $\Pr(x\text{ received}\,|\,z\text{ sent})>0$ (or both).

We are interested in determining which channels admit matched metrics,
and which metrics admit matched channels.  For example, as described
in the introduction, it is well-known that the Hamming metric and the
$n$-fold binary symmetric channel BSC($n$) are matched; the Euclidean
metric and the $n$-fold additive white Gaussian noise channel
AWGN($n$) are matched \cite{Costello}; and the Lee metric and certain $n$-fold $q$-ary
channels are matched (see \cite{ChiangWolf}, Theorem 1).

The general question of which metrics admit matched channels is much
simpler than the question of which channels admit matched metrics.
Indeed, \emph{every} metric is matched to some channel:

\begin{proposition}  
  For any finite metric space $\left(\mathcal{X}, d\right)$ there is a
  channel $W:\cX \to \cX$ matched to $d$.
\end{proposition}

\begin{proof} Given a finite metric space $(\cX,d)$, we construct a
  channel $W:\cX \to \cX$ by constructing the conditional
  probabilities $\Pr(y \, | \, x) = \Pr(y \text{ sent} \, | \, x
  \text{ received})$ for $x,y \in \cX$. 

Fix $0<\epsilon<1$.  For $x,y \in \cX$, set $\beta_{xy} = \epsilon^{d(x,y)}$, set
$\gamma_x = \sum_{y \in \cX}\beta_{xy}$, and set $\Pr(y \, | \, x) = \frac{\beta_{xy}}{\gamma_x}$.
Then $0 < \Pr(y\, | \, x) \leq 1$ and, for a fixed $x \in \cX$, 
\begin{align*}
\sum_{y \in \cX} \Pr(y \,  | \, x) &= \sum_{y \in \cX}
\frac{\beta_{xy}}{\gamma_x}\\
&= \frac{1}{\gamma_x} \sum_{y \in \cX} \beta_{xy}\\
&= \frac{1}{\gamma_x}\cdot \gamma_x\\
&= 1
\end{align*}
and so this definition yields a valid channel.  To see that this
channel is matched to our metric, let $x,y,z \in \cX$.  Then
\[
P(y \, | \, x) > P(z \, | \, x) \iff
\frac{\beta_{xy}}{\gamma_x} > \frac{\beta_{xz}}{\gamma_x} \iff
\epsilon^{d(x,y)}> \epsilon^{d(x,z)} \iff d(x,y) < d(x,z),
\]
and so condition~(\ref{newnewmatched}) is satisfied.
\end{proof}

On the other hand, not every channel has a matched metric, as the
following simple example demonstrates:

\begin{example}[\textbf{Inexistence of a matched metric}]\label{NoMatchExample}
Let $\mathcal{X}=\left\{  x,y,z\right\}  $ and $W:\cX \to \cX$ be defined by the
probabilities%
\[%
\begin{array}
[c]{ccc}%
\Pr\left(  x|x\right)  =a & \Pr\left(  x|y\right)=b & \Pr\left(  x|z\right)  =c\\
\Pr\left(  y|x\right)  =c & \Pr\left(  y|y\right)=a & \Pr\left(  y|z\right)  =b\\
\Pr\left(  z|x\right)  =b & \Pr\left(  z|y\right)=c & \Pr\left(  z|z\right)  =a
\end{array}
\]
with $a>b>c>0$ and $a+b+c=1$; for example we could have $a = \frac12$,
$b = \frac13$ and $c = \frac16$.  
Suppose $d:\cX \times \cX \to \R$ is a matched metric for $W$.  Then
\begin{align}
d(x,x)  &  <d(x,y)<d(x,z)\label{contra}\\
d(y,y)  &  <d(y,z)<d(y,x)\nonumber\\
d(z,z)  &  <d(z,x)<d(z,y).\nonumber
\end{align}
However, this leads to a contradiction:
\begin{align*}
d(z,x)  &  <d(z,y)\\
&  =d(y,z)<d(y,x)\\
&  =d(x,y)<d(x,z)\\
&  =d(z,x),
\end{align*}
where the inequalities follow from the inequalities in (\ref{contra}) and the
equalities are just the symmetry of the metric $d$.  Therefore there
can be no matched metric for the channel $W$.
\end{example}

The preceding tiny example can be generalized to any larger alphabet size
and we have the following:

\begin{proposition}
For any alphabet $\mathcal{X}$ with $\left|  \mathcal{X}\right|  \geq3$, there
is a channel $W:\cX \to \cX$ that does not admit a matched metric.
\end{proposition}

\begin{proof} If $\left|\cX\right| = 3$, proceed as in
  Example~\ref{NoMatchExample}.  Otherwise, write $\mathcal{X}
  =\mathcal{X}_0 \cup \mathcal{Y} $ where $|\mathcal{X}_0| = 3$ and
  $\left|\cY\right| = M \geq 1$.  Label the elements of $\cX_0$ so that
  $\cX_0 = \left\{ x,y,z\right\}$.  Fix positive real numbers
  $a>b>c>d$ with $a+b+c+d = 1$.  Define conditional probabilities for
  $u,v \in \cX$ by 
\[
\Pr(u \, | \, v) = 
\begin{cases}
\Pr_0(u \, | \, v) & \text{ if } u, v \in \cX_0\\
0 & \text{ if } u \in \cX_0, v \in \cY\\
\frac{d}{M} & \text{ if } u \in \cY, v \in \cX_0\\
a & \text{ if } u = v \in \cY\\
\frac{1-a}{M-1} & \text{ if } u, v \in \cY \text{ and } u \neq v,
\end{cases}
\]
where $\Pr_0(u \, | \, v)$ is as described in
Example~\ref{NoMatchExample} for $u, v \in \cX_0$.  Then it is
straightforward to check that these conditional probabilities define a
channel $W:\cX \to \cX$ and that this channel has no metric for the
same reason as in Example~\ref{NoMatchExample}.
\end{proof}
\bigskip

The 3-step cycle of Example~\ref{NoMatchExample} gives rise to a more
general obstruction criterion for the existence of a metric matching a
given channel. Given a channel $W: \cX \to \cX$, $x\in\mathcal{X}$,
and $0 \leq t \leq 1$,
we define the $t$\emph{-decision region centered at} $x$ to be
$B^{t}\left( x\right) :=\left\{ y\in\mathcal{X}\, | \, \Pr\left( x|y\right)
  \geq t\right\} $. 
We say that $x_0,x_1,...,x_{r-1}\in\mathcal{X}$ is a \emph{decision
chain} of length $r$ on $W$ if there are values
$t_0,t_1,...,t_{r-1}>0$ satisfying the following
conditions:%
\begin{align*}
\text{\emph{(FIP)} \emph{Forward Inclusion Property}: } & x_{i+1}\in B^{t_{i}}\left(  x_{i}\right)  \ \\
\text{\emph{(MEP) Backward Exclusion Property}: } &  x_{i}\notin B^{t_{i+1}}\left(  x_{i+1}\right)  \text{ }%
\end{align*}
where we consider the indices modulo $r$.  For example, taking
$x_0=x$, $x_1=y$, $x_2=z$ and $t_0=t_1=t_2=\frac{b+c}{2}$ in
Example~\ref{NoMatchExample} gives a decision chain of length 3.

With this definition we can state the following:

\begin{proposition}
Let $W$ be a channel over the alphabet $\mathcal{X}$. If $W$ admits a
decision chain of length $r \geq 3$, then there is no metric matched
to $W$.
\end{proposition}

\begin{proof}Let $x_0,x_1,...,x_{r-1}\in\mathcal{X}$ be a
decision chain on $W$ with parameters $t_0,t_1,...,t_{r-1}$ and suppose
$d:\mathcal{X}\times\mathcal{X\rightarrow}\mathbb{R}$ is a metric
matched to $W$.  
Using \emph{FIP} for $i=0$ and \emph{BEP} for $i=r-1$ we get
that
\[
d\left(  x_0,x_1\right)  <d\left(  x_0,x_{r-1}\right)  =d\left(
x_{r-1},x_0\right)
\]
where the equality follows from the symmetry property of $d$. Using
\emph{FIP} for $i=r-1$ and \emph{BEP} for $i=r-2$ we get that%
\[
d\left( x_{r-1},x_0\right)  <d\left(  x_{r-1},x_{r-2}\right)  =d\left(
x_{r-2},x_{r-1}\right).
\]
Proceeding in this manner, we get%
\[
d\left(  x_0,x_1\right)  <d\left(  x_0,x_{r-1}\right)  <d\left(
x_{r-1},x_{r-2}\right)  <\cdots<d\left(  x_{0},x_{1}\right),
\]
a contradiction.  Thus $d$ cannot exist.
\end{proof}

\bigskip

The preceding proposition gives an obstruction to the existence of a
metric matched to a channel, but many channels do not fit into this
picture. If we consider $W$ to be a channel that is both reasonable
(in the sense described above) and symmetric (in the sense that
$\Pr\left( x \text{ received} \, | \, y \text{ sent}\right) =\Pr\left(
  y \text{ received}\, |\, x \text{ sent}\right)$ for every
$x,y\in\mathcal{X}$), then defining
\[
e\left(  x,y\right)  =
\begin{cases}
0& \text{ if }x=y\\
1-\Pr\left(  x|y\right)  & \text{ if }x\neq y
\end{cases}
\]
we get that $e$ satisfies all the properties of a metric except for
the triangle inequality; as we shall see in Lemma~\ref{squeeze-lemma}
below, it is not difficult to obtain a metric $d\left( x,y\right)$
from $e\left( x,y\right)$.

It follows that the main difficulty in finding a metric matched to a
given channel is to find a \emph{symmetric} function satisfying
condition (\ref{matched}) (or, equivalently, (\ref{newmatched}) or (\ref{newnewmatched})).

In the next section, we construct a matched metric for the $n$-fold
$Z$-channel, for any $n$.  We consider this result to be a bit
surprising, since, as described above, it is the symmetry property that
poses the most difficulty in constructing a metric matched to a given
channel, and the $Z$-channel is as asymmetrical as possible, in the
sense that for $x\neq y$ we have that $\Pr\left( x \text{ received} \,
  |\, y \text{ sent}\right) >0$ implies $\Pr\left(y \text{
    received} \, |\, x \text{ sent}\right) =0$.

\section{A matched metric for the $Z$-channel\label{Z}}

The $Z$-channel is the memoryless binary input and output channel
with transition probabilities given by $\Pr(0|0)=1$, $\Pr(1|0)=0$,
$\Pr(0|1)=q$, $\Pr(1|1)=1-q$, where $0<q<\frac{1}{2}$ and, as usual,
we write $\Pr(x|y)$ to mean $\Pr(x \text{ received}\, |\, y \text{
  sent})$. The $n$-fold $Z$-channel is the memoryless channel with
input and output $\mathbb{F}_{2}^{n}$ with
\[
\Pr(x|y)=\prod_{i=1}^{n}\Pr(x_{i}|y_{i})
\]
for $x=(x_{1},\cdots,x_{n}),y=(y_{1},\cdots,y_{n})\in\mathbb{F}_{2}^{n}$.

The main result of this section is as follows:

\begin{theorem}\label{z-theorem} For any $n \geq 1$, there is a metric matched to the
  $n$-fold $Z$-channel.
\end{theorem}

Before proving Theorem~\ref{z-theorem}, we need a lemma.

\begin{lemma}\label{squeeze-lemma} Let $\cX$ be a finite set and
  suppose $e: \cX \times \cX \to \R$ is a
 \emph{semimetric}, i.e., a function satisfying
\begin{enumerate}
\item $e$ is \emph{symmetric}: $e(x,y) = e(y,x)$ for all $x,y \in
  \cX$; and
\item $e$ is \emph{nonnegative}: $e(x,y) \geq 0$ for all $x,y \in
  \cX$, with equality if and only if $x=y$
\end{enumerate}

Then there is a metric $d$ on $\cX$ such that $d(x,y)<d(x,z)$ if and
only if $e(x,y)<e(x,z)$ for every $x,y,z \in \cX$.
\end{lemma}

\begin{proof} Since $\cX$ is finite, the set $\{e(x,y) \, | \, x,y \in
  \cX, x \neq y\}$ has both a maximal and a minimal element; set $m =
  \min\{e(x,y) \, | \, x,y \in \cX, x \neq y\}$ and $M = \max\{e(x,y) \, | \,
  x,y \in \cX\}$.  Fix $\delta$ with $0 < \delta < \frac13$ and let
  $f:[m,M] \to [1-\delta,1+\delta]$ be a strictly increasing bijective
  function.  (For example, take $f$ to be the linear function which
  maps $m$ to $1-\delta$ and $M$ to $1+\delta$.)
  Define $d:\cX \times \cX \to \R$ by 
\[
d(x,y) = \begin{cases} e(x,y) = 0 & \text{ if } x=y,\\ f(e(x,y)) &
  \text{ otherwise.} \end{cases}
\]
The symmetry and nonnegativity of $d(\cdot,\cdot)$
  follow immediately from these properties of $e(\cdot,\cdot)$ and the
  fact that
$f$ is strictly increasing. To check that $d$ satisfies the
triangle inequality, let $x,y,z$ be distinct elements of $\cX$.  Then
\[
d(x,y) + d(y,z) \geq 2(1-\delta) > 2(1-\frac{1}{3})= \frac{4}{3} >
1+\delta \geq d(x,z).
\]
Hence $d$ is a metric. We remark there many other ways to transform a
semimetric into a metric. Details can be found in the
first chapter of \cite{Deza}.
\end{proof}

\begin{proof}[Proof of Theorem~\ref{z-theorem}.]  We proceed by
  induction on $n \geq 1$.   For the base case of $n=1$, we note that
  the Hamming metric is matched to the $Z$-channel on $\F_2$.

  Suppose there is a matched metric for the $n$-fold $Z$-channel,
  determined by the $2^{n}\times2^{n}$ matrix $D_{n}$. Our goal is to
  construct a $2^{n+1}\times2^{n+1}$ matrix $D_{n+1}$ that represents
  a metric that is matched to the $(n+1)$-fold $Z$-channel. For any $u
  \in\mathbb{F}_{2}^{n+1}$, we write $u = (x_{1}, \dots, x_{n},
  \theta) =: x\theta$, where $x \in\mathbb{F}_{2}^{n}$ and
  $\theta\in\mathbb{F}_{2}$; given an ordering $v_{1}, \dots, v_{N}$ of
  the elements of $\mathbb{F}_{2}^{n}$ ($N = 2^{n}$), this yields an
  ordering $v_{1}0, \dots, v_{N}0, v_{1}1, \dots, v_{N}1$ of
  $\mathbb{F}_{2}^{N+1}$.

Let $P_{n}=(P_{x,y})_{x,y\in\mathbb{F}_{2}^{n}}$ be the probability matrix for
the $n$-fold $Z$-channel, so that $P_{x,y}=\Pr(x \text{ received}\,
|\, y \text{ sent})$. Then the probability
matrix $P_{n+1}$ for the $(n+1)$-fold $Z$-channel is given by
\begin{align*}
P_{n+1}  &  =
\begin{blockarray}{ccc}
&v_10 \quad \dots \quad v_N0 & v_11 \quad \dots \quad v_N1\\
\begin{block}{c(c|c)}
v_10 &&\\
\vdots&P_n\cdot \Pr(0 \text{ received}\, | \, 0 \text{ sent})& P_n \cdot \Pr(0 \text{ received}\, | \, 1 \text{ sent})\\
v_N0 &&\\
\cline{2-3}
v_11 &&\\
\vdots&P_n\cdot \Pr(1 \text{ received}\, | \, 0 \text{ sent})& P_n \cdot \Pr(1 \text{ received}\, | \, 1 \text{ sent})\\
v_N1 &&\\
\end{block}
\end{blockarray}\\
&  =\begin{blockarray}{ccc}
&v_10 \quad \dots \quad v_N0 & v_11 \quad \dots \quad v_N1\\
\begin{block}{c(c|c)}
v_10 &&\\
\vdots&P_n& P_n \cdot q\\
v_N0 &&\\
\cline{2-3}
v_11 &&\\
\vdots&0& P_n \cdot (1-q)\\
v_N1 &&\\
\end{block}
\end{blockarray} \quad .
\end{align*}
We will use this information to construct a matrix $D_{n+1}$ that
determines a metric matched to the $(n+1)$-fold $Z$-channel. The
entries of the matrix $D_{n+1}=(d_{uv})_{u,v\in\mathbb{F}_{2}^{n+1}}$
must satisfy the following properties:
\begin{enumerate}
\item[(M)] $d$ must be \emph{matched}: $d_{uv}<d_{uw}$ if and only if
  $\Pr(u|v)>\Pr(u|w)$.
\item[(S)] $d$ must be \emph{symmetric}: $d_{uv} = d_{vu}$ for every
  $u, v\in\mathbb{F}_{2}^{n+1}$.
\item[(N)] $d$ must be \emph{nonnegative}: $d_{uv} \geq0$ for every
  $u,v \in\mathbb{F}_{2}^{n+1}$, with equality if and only if $u=v$.
\item[(T)] $d$ must satisfy the \emph{triangle inequality}:
  $d_{uv}+d_{vw} \geq d_{uw}$ for every $u,v,w\in \mathbb{F}_{2}^{n+1}$.
\end{enumerate}

Note that the last three of these properties are required for
$D_{n+1}$ to represent a metric, while the first is what makes the
metric matched to the channel. We begin by constructing a matrix $E$
that satisfies properties (M), (S), and (N), i.e., a matrix that
represents a semimetric matched to the channel. We then apply
Lemma~\ref{squeeze-lemma} to $E$ to transform $E$ into a matrix $D$
that satisfies property (T) while maintaining the other three
properties.  This modified matrix will be our desired $D_{n+1}$.

Because $E$ must be a symmetric matrix by (S), we can write
\[
E = \begin{blockarray}{ccc}
&v_10 \quad \dots \quad v_N0 & v_11 \quad \dots \quad v_N1\\
\begin{block}{c(c|c)}
v_10 &&\\
\vdots&A& B\\
v_N0 &&\\
\cline{2-3}
v_11 &&\\
\vdots&B^{T}&C\\
v_N1 &&\\
\end{block}
\end{blockarray} \quad ,
\]
where we require $A = (a_{xy})_{x,y \in\mathbb{F}_{2}^{n}}$, $B =
(b_{xy})_{x,y \in\mathbb{F}_{2}^{n}}$ and $C = (c_{xy})_{x,y \in\mathbb{F}%
_{2}^{n}}$ to be $2^{n} \times2^{n}$ matrices, with $A$ and $C$ symmetric.

\textbf{Determining matrix $A$:} To satisfy (M), we must have $a_{xy}<a_{xz}$ if and only if $\Pr
(x0|y0)>\Pr(x0|z0)$ if and only if $\Pr(x|y)>\Pr(x|z)$. Thus $A$ must
represent a matched metric for the $n$-fold $Z$-channel, and we may
set $A=D_{n}$.

\textbf{Determining matrix $B$:} Let us consider the entry $b_{x,y}$ of the matrix $B$. We break this into
three cases.

\textit{Case 1. } First, suppose $\Pr(x|y)\neq0$ and $y$ is not the
all-ones vector in $\mathbb{F}_{2}^{n}$.  Without loss of generality,
we may assume that all of the 0's in $y$ are at the beginning, so that
$y=0^{j}1^{k}$ with $j\geq1$ and $j+k=n$, where, for example, we mean by $0^21^3$
the vector $(0,0,1,1,1)$.  Since $\Pr(1|0)=0$, the first $j$ coordinates of $x$
are 0 as well. Hence, without loss of generality, we may assume that
$x=0^{j}0^{s}1^{t}$ with $s+t=k$. Now set
$z=(1,y_{2},\dots,y_{n})$. Then
\[
\Pr(x|z)=\Pr(0|1)\prod_{i=2}^{n}\Pr(x_{i}|y_{i})=q\Pr(x|y)
\]
since $\Pr(x_{1}|y_{1})=\Pr(0|0)=1$. Thus we have
\[
\Pr(x0|y1)=q\Pr(x|y)=\Pr(x|z)=\Pr(x0|z0)
\]
and so we set $b_{xy}=a_{xz}$. We remark that, since $x\neq z$, the induction
hypothesis ensures that $a_{xz}\neq0$, and hence also $b_{xy}\neq0$.

\textit{Case 2. } We now consider the case where $y = 1^n$
is the all-ones vector in $\F_2^n$, and find the value of
$b_{x1^n}$.  Note first that $\Pr(1^n0|1^n1) =
q(1-q)^n$ is the second-largest entry in the row of $P_{n+1}$ indexed
by $1^n0$, second only to $\Pr(1^n1|1^n1)$.
Since $a_{1^n1^n} =0$, we therefore require $b_{1^n1^n}$ to be
smaller than every nonzero $a_{1^nz}$; for concreteness, we set
$b_{1^n1^n} = \alpha\min\{a_{1^nz}|z \neq 1^n\}$ with $0<\alpha<1$.

For $x \neq 1^n$, without loss of generality, we may assume
$x=0^{s}1^{t}$, where $s\geq 1$ and $s+t=n$. Then
\[
\Pr(x|1^n)=q^{s}(1-q)^{t}%
\]
and
\[
\Pr(x0|1^n1)=q^{s+1}(1-q)^{t}.
\]
Suppose $z\neq 1^n$ satisfies $\Pr(x|z)\neq0$; since $x \neq
1^n$, we know such a $z$ exists.  Since $\Pr(1|0)=0$, without
loss of generality we can write $z=0^{j}1^{k}1^{t}$, where $j+k=s$ and
$j\geq1$. This means
\[
\Pr(x|z)=\Pr(0|0)^{j}\Pr(0|1)^{k}\Pr(1|1)^{t}=q^{k}(1-q)^{t},
\]
where $k<s$. Putting this together, we have
\[
\Pr(x0|1^n1)=q^{s+1}(1-q)^{t}<q^{k+1}(1-q)^{t}=\Pr(x0|z1)<q^{k}(1-q)^{t}%
=\Pr(x0|z0)
\]
and so we require $b_{x1^n}>b_{xz}>a_{xz}$ for every $z\neq 1^n$ with
$\Pr(x|z)\neq0$. For concreteness, we set
\[
b_{x1^n}=\beta\max\{b_{xz}\,|\,z\neq 1^n\text{ and }\Pr(x|z)\neq0\}
\]
with $\beta>1$.

\textit{Case 3. } Finally, if $\Pr(x|y)=0$, then $\Pr(x0|y1)=0$ and so
\[
\Pr(x0|y1)<\Pr(x0|z1)=q\Pr(x0|z0)<\Pr(x0|z0)=\Pr(x|z)
\]
for every $z$ with $\Pr(x|z)\neq0$. This means we require
\[
b_{xy}>b_{xz}>a_{xz}%
\]
for every $z$ with $\Pr(x|z)\neq0$, and we set
\[
b_{xy}=\gamma\max\{b_{xz}\,|\,\Pr(x|z)\neq0\}
\]
with $\gamma>1$.

By the remark made in Case 1 and the constructions in the two
remaining cases, the matrix $B=(b_{xy})$ has strictly positive
entries.

\textbf{Determining matrix $C$:} Because $\Pr(x1|y0)=0$ for all $x$
and $y$, for any $x$ and $z$ with $\Pr(x|z) \neq 0$, we must have $c_{xz}<b_{yx}$ for
all $y$. Because $\Pr(x1|y1)<\Pr(x1|z1)$ if and only if
$\Pr(x|y)<\Pr(x|z)$, the matrix $C$ must represent a matched metric
for the $n$-fold $Z$-channel. By choosing $\delta$ sufficiently small and setting
$C=\delta D_{n}$, we can satisfy these conditions.

We now have a matrix
\[
E = (e_{uv})_{u,v \in\mathbb{F}_{2}^{n+1}} =
\begin{pmatrix}
A & B\\
B^{T} & C
\end{pmatrix}
\]
that satisfies properties (M), (S) and (N) above.  Using
Lemma~\ref{squeeze-lemma}, we can transform $E$ in such a
way that we force the triangle inequality (T) to hold without
affecting the other three properties.  Thus the resulting matrix
$D_{n+1}=(d_{uv})_{u,v\in \mathbb{F}_{2}^{n+1}}$ represents a master
metric for the $(n+1)$-fold $Z$-channel, as desired.
\end{proof}

\section{Asymmetric channels}\label{bac}

A binary asymmetric channel (BAC) with parameters $\left( p,q\right) $
is a memoryless channel with binary input and output alphabet with
transition probabilities given by $\Pr(0|0)=1-p$, $\Pr(1|0)=p$,
$\Pr(0|1)=q$, $\Pr(1|1)=1-q$, where $0< p< q<\frac{1}{2}$. The
boundary cases of a BAC are the BSC (for
$p=q$) and the $Z$-channel (for $p=0$).  The squeezing function of
Lemma~\ref{squeeze-lemma} ensures the triangle inequality can always
be attained.  Hence the unique difficulty to construct a metric
matched to a given channel lies on the symmetry of a distance
matrix. From this point of view, we could expect that finding a
matched metric for an asymmetric channel should become harder as the
asymmetry of the channel grows, that is, as $p$ becomes closer to $0$
and we get a $Z$-channel. On the other hand, as shown in
Section~\ref{Z} above, the $Z$-channel does admit a matched metric.
We remark that the well known asymmetric distance (see, for example,
\cite{SerbanRao}) is a metric but it is not matched to the BAC.

We conjecture that there is a matched metric for any BAC. We briefly
describe some approaches to the problem.

We first note that the decision rule for the $Z$-channel of length $n$
with $P(0|1) = q$, $0<q<\frac12$, is independent of $q$, as is the
decision rule for the BSC of length $n$ with $P(1|0) =q = P(0|1)$,
$0<q <\frac12$.  However, there are in general multiple decision rules
for BACs of length $n$, depending on the relationship between $p =
P(1|0)$ and $q = P(0|1)$, $0<p<q<\frac12$.  For example, while the
decision rule for the BAC of length $2$ is independent of $p$ and $q$,
$0<p<q<\frac12$, there are two distinct decision rules for BACs of
length $3$.  Explicitly, the decision rules for the BAC with $p = 0.1$
and $q = 0.2$ and the BAC with $p = 0.1$ and $q = 0.4$ do not coincide.
Thus, the question of whether there is a matched metric for any BAC
goes beyond just the length of the channel; the specific values of $p$
and $q$ must be considered as well.

One approach to proving that every BAC has a matched metric, which
proves to be unsuccessful, uses the fact that we may think of there
being a continuum between the $Z$-channel and the binary symmetric
channel, by way of binary asymmetric channels. We know that, for any
$n$, the $Z$-channel admits a matched metric (as proven in
Section~\ref{Z}) and the BSC admits a matched metric (the Hamming
metric on $\F_2^n$).  Letting $D_n$ be the matrix representing the
metric matched to the $Z$ channel of length $n$ and $H_n$ be the
matrix representing the Hamming metric on $\F_2^n$, we set $E_n(t) =
(1-t)D_n + t H_n$.  Then $E_n(0) = D_n$, $E(1) = H_n$, and, letting
$\Delta_n(t)$ be the matrix obtained by applying Lemma~\ref{squeeze-lemma} to
$E_n(t)$, one might
hope that the interval $0<t<1$ breaks into subintervals on which
$\Delta_n(t)$ is a matrix representing a metric matched to the various BACs of
length $n$.  

This approach works for $n=2$: we find that if we take $0<\alpha =
\delta<1$, $\beta>1$, and $\gamma>1$ in the construction of
Section~\ref{Z}, then there is a non-empty subinterval of $(0,1)$,
depending only on the relationship between $\alpha$ and $\gamma$, such
that $\Delta_2(t)$ is matched to the (unique) BAC of length $2$ for every $t$
in that subinterval.  In particular, if $\gamma \geq \frac1\alpha$,
then $\Delta_2(t)$ is matched to the unique BAC of length $2$ for every $t$
with $0<t<1$.

The approach fails, however, for $n=3$.  Recall that the metric $D_3$
matched to the $Z$-channel of length 3 is constructed in
Section~\ref{Z} by applying Lemma~\ref{squeeze-lemma} to a matrix $E$
that depends on the matrix $D_2$ (previously produced by the
construction) and constants $\alpha$, $\beta$, $\gamma$ and $\delta$.
Because there is no a priori reason that these constants be the same
in the length 2 and length 3 constructions and we wish to preserve as
much flexibility as possible, we write $\alpha_i$, $\beta_i$,
$\gamma_i$ and $\delta_i$ for the constants used in the construction
of the matrix $D_i$ for $i=2,3$.  We find that there are no values of
these constants such that the metric represented by the matrix
$\Delta_3(t)$ obtained by applying Lemma~\ref{squeeze-lemma} to
$E_3(t) = (1-t)D_3 + t H_3$ represents either of the two BAC channels
of length $3$ for any value of $t$.

A different approach is to try to construct a matrix representing a
metric matched to the BAC directly.  First, construct the probability
matrix for the BAC under consideration.  This yields an integer ``distance
order matrix'' whose rows give the reverse order of the sizes of the
entries of the probability matrix.  For example, when $n=2$, the
probability and distance order matrices for the unique BAC with
$0<p<q<\frac12$ are 
\[
\begin{blockarray}{ccccc}
&00&01&10&11\\
\begin{block}{c(cccc)}
00&(1-p)^2 & (1-p)q& (1-p)q&q^2\\
01&(1-p)p & (1-p)(1-q) & pq & (1-q)q\\
10&(1-p)p & pq & (1-p)(1-q) &  (1-q)q\\
11&p^2 & p(1-q) & p(1-q) & (1-q)^2\\
\end{block}
\end{blockarray}
\text{\quad and \quad}
\begin{blockarray}{ccccc}
&00&01&10&11\\
\begin{block}{c(cccc)}
00&1&2&2&3\\
01&3&1&4&2\\
10&3&4&1&2\\
11&3&2&2&1\\
\end{block}
\end{blockarray}
\]
We now work line-by-line to construct a symmetric matrix in which all
diagonal entries are 0 and all other entries are positive.
Considering only the first line, we have
\[
\begin{pmatrix}
0 & a & a & b\\
a & 0 &    &   \\
a &    & 0 &   \\
b &    &    & 0
\end{pmatrix}
\]
with $a<b$.  Moving on to the second and third lines gives 
\[
\begin{pmatrix}
0 & a & a & b \\
a & 0 & c & d \\
a & c & 0 & e \\
b & d & e & 0
\end{pmatrix}
\]
with $d<a<c$ and $e<a<c$.  Finally, looking at the fourth line we see
that we must have $d<b$ and $d=e$.  Putting this together yields a matrix
\[
\begin{pmatrix}
0 & a & a & b \\
a & 0 & c & d \\
a & c & 0 & d \\
b & d & d & 0
\end{pmatrix}
\]
with $d<a<b$, $d<a<c$, and $b$ and $c$ incomparable.  This means that,
for example, applying Lemma~\ref{squeeze-lemma} to the matrix
\[
\begin{pmatrix}
0 & 2 & 2 & 3 \\
2 & 0 & 3 & 1 \\
2 & 3 & 0 & 1 \\
3 & 1 & 1 & 0
\end{pmatrix}
\]
yields a matrix representing a metric matched to the unique BAC of length 2.  A similar
procedure gives that applying Lemma~\ref{squeeze-lemma} to the
matrices
\[
\begin{pmatrix}
0 & 3 & 3 & 4 & 3 & 4 & 4 & 5 \\
 3 & 0 & 5 & 2 & 5 & 2 & 6 & 4 \\
 3 & 5 & 0 & 2 & 5 & 6 & 2 & 4 \\
 4 & 2 & 2 & 0 & 6 & 3 & 3 & 1 \\
 3 & 5 & 5 & 6 & 0 & 2 & 2 & 4 \\
 4 & 2 & 6 & 3 & 2 & 0 & 3 & 1 \\
 4 & 6 & 2 & 3 & 2 & 3 & 0 & 1 \\
 5 & 4 & 4 & 1 & 4 & 1 & 1 & 0 \\
\end{pmatrix}
\text{\phantom{space} and \phantom{space}}
\begin{pmatrix}
 0 & 4 & 4 & 5 & 4 & 5 & 5 & 6 \\
 4 & 0 & 5 & 2 & 5 & 2 & 6 & 3 \\
 4 & 5 & 0 & 2 & 5 & 6 & 2 & 3 \\
 5 & 2 & 2 & 0 & 6 & 3 & 3 & 1 \\
 4 & 5 & 5 & 6 & 0 & 2 & 2 & 3 \\
 5 & 2 & 6 & 3 & 2 & 0 & 3 & 1 \\
 5 & 6 & 2 & 3 & 2 & 3 & 0 & 1 \\
 6 & 3 & 3 & 1 & 3 & 1 & 1 & 0 \\
\end{pmatrix}
\]
yields matrices representing metrics matched to the two BACs of length
3.

Thus we see that every BAC of length 2 or 3 has a matched metric, and
we conjecture that this holds for general lengths.

\section*{Acknowledgements} Example~\ref{NoMatchExample} is due to the
first author's student, Rafael Gregorio.  This work began when both
authors were visiting Centre Interfacultaire Bernoulli at \'Ecole
Polytechnique F\'ed\'erale de Lausanne, and they thank the CIB for its
hospitality.  The first author was partially supported by grant 2013/25977-7, S\~{a}o Paulo Research Foundation (Fapesp). The second author was supported in part by NSF Grant
\#0903517.


\bibliographystyle{amsplain}
\bibliography{refs}









\end{document}